\documentclass[preprint,12pt]{elsarticle}
\usepackage{mathrsfs,graphicx}
\usepackage{amsmath,amsfonts,amssymb,amsthm}

\newcommand{\ie}{{i.e.,~}}

\newcommand{\N}{\mathbb{N}}
\newcommand{\D}{\mathscr{D}}
\newcommand{\p}{{\ensuremath{\cal P}}}
\newcommand{\np}{{\ensuremath{\cal NP}}}

\newtheorem{thm}{Theorem}
\newtheorem{prop}{Proposition}
\newtheorem{lem}{Lemma}
\newtheorem{obs}{Observation}

\newcommand{\ignore}[1]{}

\journal{Information Processing Letters}

\begin{document}

\begin{frontmatter}

\title{Topological Additive Numbering of Directed Acyclic Graphs\tnoteref{grant}}

\author[a]{Javier Marenco}
\author[a]{Marcelo Mydlarz}
\author[b]{Daniel Sever\'in\fnref{correspon}}

\address[a]{ Universidad Nacional de General Sarmiento,
	Argentina }

\address[b]{ Universidad Nacional de Rosario,
	Argentina }

\tnotetext[grant]{Partially supported by ANPCyT PICT-2009-0119 (Argentina)}
\fntext[correspon]{Corresponding author at Depto. de Matem\'atica, Escuela de Formaci\'on B\'asica, Facultad de Ciencias Exactas y Naturales, Universidad
Nacional de Rosario. \emph{Address}: Pellegrini 250, Rosario, Argentina\\
\emph{E-mail addresses}: \texttt{\{jmarenco,mmydlarz\}@ungs.edu.ar},
\texttt{daniel@fceia.unr.edu.ar}}

\begin{abstract}
We propose to study a problem that arises naturally from both Topological Numbering of Directed Acyclic Graphs, and Additive Coloring (also known as Lucky Labeling).
Let $D$ be a digraph and $f$ a labeling of its vertices with positive integers; denote by $S(v)$ the sum of labels over all neighbors of each vertex $v$.
The labeling $f$ is called \emph{topological additive numbering} if $S(u) < S(v)$ for each arc $(u,v)$ of the digraph.
The problem asks to find the minimum number $k$ for which $D$ has a topological additive numbering
with labels belonging to $\{ 1, \ldots, k \}$, denoted by $\eta_t(D)$.

We characterize when a digraph has topological additive numberings,
give a lower bound for $\eta_t(D)$, and
provide an integer programming formulation for our problem, characterizing when its coefficient matrix is totally unimodular.
We also present some families for which $\eta_t(D)$ can be computed in polynomial time.
Finally, we prove that this problem is \np-Hard even when its input is restricted to planar bipartite digraphs.
\end{abstract}


\end{frontmatter}

%%%%%%%%%%%%%%%%%%%%%%%%%%%%%%%%%%%%%%%%%%%%%%%%%%%%%%%%%%%%%%%%%%%%%%%%%%%%%%%%

\section{Introduction} \label{SINTRO}

\emph{Graph Coloring} (GC) is one of the most representative problems in graph theory and combinatorial optimization because of its
practical relevance and theoretical interest. Below, we present two known variants of GC.

Let $D = (V, A)$ be a directed acyclic graph (DAG), and let $S: V \rightarrow \N$ be a labeling of the vertices of $D$.
If $S(u) < S(v)$ for every $(u, v) \in A$, then $S$ is called a \emph{topological numbering} of $D$~\cite{TOPOBOOK}.
We refer to the problem of finding the minimum number $k$ for which such labeling $S$ satisfies $S(v) \le k$ for all $v \in V$ as \emph{Topological Numbering} of DAGs (TN).
This number $k$ is also the size of the largest directed path in $D$ (Gallai Theorem~\cite{GALLAI}).
TN is solvable in polynomial time and generalizations of it give rise to different applications: PERT/CPM problems and the
buffer assignment problem for weighted rooted graphs~\cite{BALANCEDACYCLIC}, and frequency assignment problems with fixed orientations~\cite{GROTSCHEL}.

The other variant of GC in which we are interested is \emph{Additive Coloring} (AC), also known as Lucky Labeling.
Let $G = (V, E)$ be a graph, $f: V \rightarrow \N$ a labeling of its vertices and $S(v)$ the sum of labels over all
neighbors of $v$ in $G$, \ie $S(v) = \sum_{w \in N(v)} f(w)$, where $N(v)$ is the set of neighbors of $v$.
If $S(u) \neq S(v)$ for every $(u,v) \in E$, then $f$ is called \emph{additive $k$-coloring} of $G$, where
$k$ is the largest label used in $f$.
AC consists in finding the \emph{additive chromatic number} of $G$, which is defined as the least number $k$ for which $G$ has an additive $k$-coloring
and is denoted by $\eta(G)$.

AC was first presented by Czerwi\'nski, Grytczuk and Zelazny~\cite{LUCKYORIGINAL}.
They conjecture that $\eta(G) \leq \chi(G)$ for every graph $G$, where $\chi(G)$ is the
chromatic number of $G$.
The problem as well as the conjecture have recently gained considerable interest~\cite{LUCKYCOMPLEXITY,LUCKYCOTASUP,ADDITIVEPLANAR}.

In particular, we proposed an exact algorithm for solving AC based on Benders' Decomposition~\cite{MACI2013}.
This algorithm needs to solve several instances of an ``oriented version'' of AC.
Let $D = (V, A)$ be a DAG, $f: V \rightarrow \N$
a labeling and $S(v) = \sum_{w \in N(v)} f(w)$ for all $v \in V$.
If $S(u) < S(v)$ for every $(u, v) \in A$, then $f$ is called \emph{topological additive $k$-numbering} of $D$, with
$k$ the largest label used in $f$.

Unlike other coloring problems (including AC and TN), a digraph may lack any topological additive numbering.
Let $\D$ denote the set of digraphs that have at least one topological additive numbering.
Then, for $D \in \D$, the \emph{topological additive number} of $D$, denoted by $\eta_t(D)$, is defined as the least number $k$ for which $D$ has a topological additive $k$-numbering.
We call the problem of finding this number \emph{Topological Additive Numbering} of DAGs (TAN).

As far as we know, there are no references to TAN in the literature.
Our main contribution is to address TAN from a computational point of view.
We first present some properties of TAN, including a lower bound for $\eta_t(D)$ and families of digraphs for which it is easy to exactly compute this number.
We also give a linear integer programming formulation of TAN and characterize when its coefficient matrix is totally unimodular.
At the end, we show that the problem is \np-Hard even for planar bipartite digraphs.

\section{Basic properties of TAN} \label{BASIC}

Let $D = (V, A)$ be a DAG with $V = \{1,\ldots,n\}$.
We will assume that $D$ is connected,
and its vertices are ordered so that $u < v$ holds whenever $(u, v) \in A$.
As usual, $d(v)$ denotes the degree of vertex $v \in V$, and $G(D)$ the undirected underlying graph of $D$.

We first note that $\eta_t(D) \geq \eta(G(D))$. Therefore,
lower bounds for the additive chromatic number also hold for the topological additive number. For instance, in~\cite{LUCKYCOTAINF} it is
proved that $\eta(G(D)) \geq \lceil \omega/(n-\omega+1) \rceil$, where $\omega$ is the size of a maximum clique of $G(D)$.
However, it is possible to get a tighter bound for $\eta_t$ as follows.

\begin{prop} \label{COTAINF}
Let $D \in \D$, $Q$ a clique of $D$ and $q_F$, $q_L$ the smallest and largest vertices of $Q$ respectively.
Then,
\[
  \eta_t(D) \geq \biggl\lceil \dfrac{d(q_F)+1}{d(q_L)-|Q|+2} \biggr\rceil.
\]
\end{prop}
\begin{proof}
We follow~\cite{LUCKYCOTAINF}. Let $f$ be a topological additive $k$-numbering of $D$. For each vertex $q \in Q$, let
$Y_q = \sum_{w \in N(q) \setminus Q} f(w) - f(q).$
It is clear that $|N(q) \setminus Q| - k \leq Y_q \leq k |N(q) \setminus Q| - 1$.

On the other hand, for any $q_1, q_2 \in Q$ such that $q_1 < q_2$, we have $S(q_1) < S(q_2)$, or equivalently,
\[
  Y_{q_1} + \sum_{w \in Q} f(w) < Y_{q_2} + \sum_{w \in Q} f(w).
\]
Hence, $Y_{q_1} < Y_{q_2}$. Since $q_F \leq q \leq q_L$ for all $q \in Q$, the values of $Y_q$ must be between
$|N(q_F) \setminus Q| - k$ and $k |N(q_L) \setminus Q| - 1$. By the pigeonhole principle,
we obtain $|Q| \leq k |N(q_L) \setminus Q| - |N(q_F) \setminus Q| + k$. Therefore, $k \geq \lceil (d(q_F)+1)/(d(q_L)-|Q|+2) \rceil$.
\end{proof}
Note that
\textrm{(i)}
  this bound is tight for $D \in \D$ when $G(D)$ is a complete graph or a complete bipartite graph, and
\textrm{(ii)}
  unlike the result given in~\cite{LUCKYCOTAINF}, larger cliques do not necessarily lead to better lower bounds.

Now, we analyze when a digraph has topological additive numberings.
The following is a sufficient condition.
\begin{obs} \label{CONDINFTY}
Let $D$ be a DAG and $u$, $v$ two vertices of $D$ such that $N(u) \subseteq N(v)$.
If there is a directed path from $v$ to $u$, then $D \notin \D$.
\end{obs}
The previous condition is not necessary since the digraph in Figure~\ref{CONTRAEJEMPLO} does not belong to $\D$ either.
\begin{figure}[h]
\begin{center}
\includegraphics[scale=0.7]{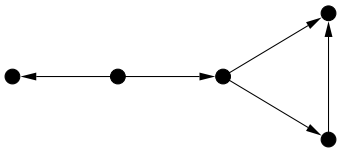}
\end{center}
\caption{A digraph that does not belong to $\D$.}
\label{CONTRAEJEMPLO}
\end{figure}

Although we do not know a combinatorial characterization of $\D$, we now describe a polynomial-time procedure
that determines whether a digraph is in $\D$.
Observe that the following integer linear program solves TAN:
\begin{align}
\min k & & \notag \\
\textrm{subject to} & & \notag \\
 & \sum_{w \in N(v)} f(w) - \sum_{w \in N(u)} f(w) \geq 1, & \forall~~(u, v) \in A \label{RARCO}\\
 & k - f(v) \geq 0, & \forall~~v \in V  \label{RKK}\\
 & f(v) \in \N, & \forall~~v \in V
\nonumber
\end{align}
We call \emph{IPF} this formulation and \emph{LR} its linear relaxation, \ie the linear program that comprises constraints~\eqref{RARCO},
\eqref{RKK} and $f(v) \geq 1$ for all $v \in V$.
If LR is infeasible, then $D \notin \D$. Otherwise, there exists an optimal solution of LR whose components are rational numbers; by multiplying these components by a suitable positive integer, we obtain a topological additive numbering of $D$.
Therefore, LR is feasible if, and only if, $D \in \D$.
Since deciding whether LR is feasible can be computed in polynomial time, we conclude that:
\begin{prop}
Given a DAG $D$, deciding whether $D \in \D$ is in \p.
\end{prop}

Since the matrix of coefficients of a standard integer programming formulation for TN is totally unimodular for every digraph~\cite{BALANCEDACYCLIC},
TN can be solved in polynomial time.
It is only natural to ask which digraphs attain such a property for TAN.
The following result shows that TAN is much harder.
\begin{thm}
Let $D$ be a connected DAG.
The matrix of IPF is totally unimodular if, and only if, $G(D)$ is a complete graph.
\end{thm}
\begin{proof}
Let $M$ be the matrix of IPF.\\
$\Leftarrow$)
Since for every $u, v \in V$ we have $N(u) \setminus N(v) = \{v\}$, constraints~\eqref{RARCO} are $f(u) - f(v) \geq 1$ for all $u < v$.
Then, $M$ has two non-zero coefficients in each row: one is $1$ and the other is $-1$. According to Prop.~2.6 of~\cite{NEMHAUSER}, $M^T$ is totally unimodular.
Therefore, $M$ is totally unimodular by Prop.~2.1 of~\cite{NEMHAUSER}.\\
$\Rightarrow$)
Suppose that $G(D)$ is not a complete graph. Since $D$ is connected, there exist $u, v, w \in V$ such that $u$ is adjacent to $v$, $v$ is adjacent
to $w$ and $u$ is not adjacent to $w$.

Consider first the case when $(u, v) \in A$.
Then, its corresponding constraint~\eqref{RARCO} has coefficients $1$ for $u$ and $w$ (and $-1$ for $v$). Let $i$ be the row index of that constraint.
Let $M'$ be the submatrix of $M$ whose columns correspond to variables $k$, $f(u)$ and $f(w)$, and whose rows are given by $i$ and constraints $k - f(u) \geq 0$,
$k - f(w) \geq 0$. Hence,
$$M' = \left( \begin{array}{ccc}
  0 & 1 & 1 \\
  1 & -1 & 0 \\
  1 & 0 & -1
\end{array} \right).$$
Since the determinant of $M'$ is $2$, $M$ is not totally unimodular.
The other case, $(v, u) \in A$, inverts the sign of the first row of $M'$, with same conclusion.
\end{proof}

Next, we present some families of digraphs where TAN is solved in polynomial time.
We say that a digraph $D$ is \emph{$r$-partite} when $G(D)$ is $r$-partite,
and $D$ is \emph{complete} when $G(D)$ is complete.
We say that an $r$-partite digraph is \emph{monotone} when it can be partitioned into $V_1, V_2, \ldots, V_r$ and
each of the arcs in $V_i \times V_j$ satisfies $i < j$.
It is easy to see that a complete $r$-partite digraph belongs to $\D$ if, and only if, it is monotone. In this case, the topological
additive number can be computed as follows.
\begin{prop}
\label{prop:monotone.rpartite}
Let $D$ be a complete monotone $r$-partite digraph. Then,
\[
  \eta_t(D) = \max \biggl\{ \biggl\lceil \dfrac{s_i}{|V_i|} \biggr\rceil : i = 1,\ldots,r \biggr\},
\]
where $s_r = |V_r|$ and $s_i = \max\{1 + s_{i+1}, |V_i|\}$ for all $i = 1,\ldots,r-1$.
\end{prop}
\begin{proof}
For any labeling $f$ and set $S \subset V$, let $f(S) = \sum_{v \in S} f(v)$.
Note that $f$ is a topological additive numbering if, and only if, $f(V_i) > f(V_{i+1})$ for all
$i = 1, \ldots, r-1$, since for all $j > i$, $u \in V_i$ and $w \in V_j$, we have $S(w) - S(u) = f(V_i) - f(V_j) > 0$.

Now, consider a labeling $f$ such that, for all $i = 1, \ldots, r$, $f(V_i) = s_i$ and $f(v) \in \{ \lfloor s_i/|V_i| \rfloor, \lceil s_i/|V_i| \rceil \}$ for all $v \in V_i$. Clearly, it is a topological additive $p$-numbering with
$p = \max \{ \lceil s_i/|V_i| \rceil : i = 1,\ldots,r-1 \}$.

In order to prove that $f$ is optimal, and by way of contradiction, suppose that there is a topological additive numbering $f'$ such that $f'(V_j) < f(V_j)$ for some $j \in \{1, \ldots, r\}$; moreover, assume that $j$ is the largest index satisfying this inequality.
Then, from $f'(V_j) \ge |V_j|$ follows that
\[
  f'(V_j) < f(V_j) = 1 + s_{j+1} = 1 + f(V_{j+1}) \le 1 + f'(V_{j+1}),
\]
contradicting that $f'(V_j) > f'(V_{j+1})$.
\end{proof}

We now extend Proposition~\ref{prop:monotone.rpartite} for monotone (not necessarily complete) bipartite digraphs.
As implied by Theorem~\ref{thm:np.hard} (in Section~\ref{COMPLEXITY}), it is $\np$-hard to obtain $\eta_t(D)$ for general
bipartite digraphs.
\begin{prop}
Let $D$ be a monotone bipartite digraph. Then,
\[
  \eta_t(D) = \max \biggl\{
	\biggl\lfloor \frac{d(u)}{d(v)} \biggr\rfloor + 1:
	v \in V_2, u \in N(v)
  \biggr\}.
\]
\end{prop}
\begin{proof}
Let $v^* \in V_2$ and $u^* \in N(v^*)$ be such that $\lfloor d(u^*)/d(v^*) \rfloor$ is maximized, and let
$p = \lfloor d(u^*)/d(v^*) \rfloor + 1 = \lceil (d(u^*) + 1)/d(v^*) \rceil$.
Proposition~\ref{COTAINF} applied to $Q = \{ u^*, v^*\}$ grants $\eta_t(D) \ge p$.
A topological additive $p$-numbering $f$, defined by $f(v) = 1$ for vertices $v \in V_2$ and $f(v)=p$ for $v \in V_1$, provides the matching upper bound.
\end{proof}

\section{Computational complexity of TAN}
\label{COMPLEXITY}
We have seen that deciding whether $D \in \D$ can be done in polynomial time.
Moreover, deciding whether $\eta_t(D) = 1$ can be computed fast by checking whether $d(u) < d(v)$ for every arc $(u,v)$.
Nevertheless, deciding whether $\eta_t(D) = 2$ is \np-complete.
The proof given below shares the same approach of~\cite{LUCKYCOMPLEXITY}.

Let $\Phi$ be a 3-SAT formula with sets of clauses $C$ and variables $X$;
let $G_\Phi = (V_\Phi, E_\Phi)$ be the graph of $\Phi$, where $V_\Phi = C \cup X \cup \{\lnot x: x \in X\}$ and
$E_\Phi = \{(x, \lnot x): x \in X\} \cup \{(c, y), (c, z), (c, w): c \in C, c = y \lor z \lor w\}$.
It is known that, given a 3-SAT formula $\Phi$ for which $G_\Phi$ is planar, deciding whether there is a truth assignment
that satisfies $\Phi$ is \np-complete~\cite{SATPLANAR}.
This problem is called \emph{Planar 3-SAT (type 2)} (P3SAT2).
We will assume, without loss of generality, that no literal is repeated within a clause (since, for instance, each clause of the form $y \lor y \lor z$ may be replaced by two clauses $x \lor y \lor z$ and $\lnot x \lor y \lor z$, where $x$ is an unused literal, maintaining planarity).

Our proof relies on a polynomial-time reduction from P3SAT2 to TAN.
Consider an instance $\Phi$ of P3SAT2 and construct the following digraph $D_\Phi$
from $G_\Phi$ as follows (Figure~\ref{CONSTR}):
\begin{itemize}
\item
  For each $x \in X$,
  add vertices $x^1, x^2, \ldots, x^5, u^1, u^2, \ldots, u^6$ to $V$,
  and replace edge $(x, \lnot x)$ with
arcs $(x^1, x)$, $(x^1, \lnot x)$, $(x^2, x^1)$, $(x^3, x^2)$, $(x^4, x^2)$, $(x^5, x^2)$,
$(u^1, x)$, $(u^2, x)$, $(u^3, x)$, $(u^4, \lnot x)$, $(u^5, \lnot x)$, $(u^6, \lnot x)$.
\item
  For each $c = y \lor z \lor w \in C$,
  add vertices $c^1, c^2, \ldots, c^5$ to $V$,
  and replace edges $(c, y)$, $(c, z)$ and $(c, w)$
  with arcs $(c, y)$, $(c, z)$, $(c, w)$, $(c, c^1)$, $(c^2, c^1)$, $(c^3, c^1)$, $(c^4, c^1)$, $(c^5, c)$.
\end{itemize}
By construction and since $G_\Phi$ is planar, $G(D_\Phi)$ is planar and bipartite.
\begin{figure}
\begin{center}
\includegraphics[scale=0.7]{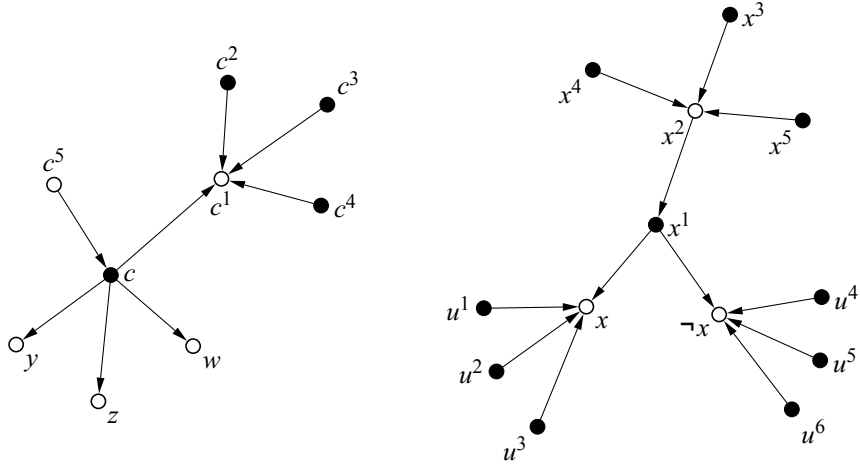}
\vspace{-13pt}
\end{center}
\caption{Construction of digraph $D_\Phi$: for each variable $x$, $D_\Phi$ has a copy of the right digraph, and for each clause $c = y \lor z \lor w$, $D_\Phi$ has a copy
of the left digraph. A bipartition is shown through the color of the vertices.}
\label{CONSTR}
\end{figure}

For the next two lemmas assume that $D_\Phi$ has a topological additive 2-numbering $f$.
\begin{lem} \label{LEMA1O}
  $f(x) + f(\lnot x) \geq 3$ for all $x \in X$.
\end{lem}
\begin{proof}
In first place, $S(x^2) < S(x^1)$. Since $x^2$ has 4 neighbors, $S(x^2) \geq 4$ and then $S(x^1) \geq 5$.
Since $S(x^1) = f(x) + f(\lnot x) + f(x^2)$ and $f(x^2) \leq 2$, we get $f(x) + f(\lnot x) \geq 3$.
\end{proof}

\begin{lem} \label{LEMA2O}
  $f(y) + f(z) + f(w) \leq 5$ for all $c = y \lor z \lor w \in C$.
\end{lem}
\begin{proof}
In first place, $S(c) < S(c^1)$. Since $c^1$ has 4 neighbors, $S(c^1) \leq 8$.
Hence, $S(c) \leq 7$. Since $S(c) = f(y) + f(z) + f(w) + f(c^1) + f(c^5)$ and $f(c^1) + f(c^5) \geq 2$, we get
$f(y) + f(z) + f(w) \leq 5$.
\end{proof}

\begin{thm}
\label{thm:np.hard}
It is \np-complete to decide whether $\eta_t(D) = 2$ for a digraph $D$ whose underlying graph is planar and bipartite.
\end{thm}
\begin{proof}
We follow~\cite{LUCKYCOMPLEXITY}.
Let $\Phi$ be a 3-SAT formula such that $G_\Phi$ is planar, and $D_\Phi$ the digraph generated from $G_\Phi$ with the procedure given above.
We only need to show that there exists a topological additive $2$-numbering $f$ of $D_\Phi$ \emph{if and only if} there also exists a truth assignment
$\Gamma: X \rightarrow \{true, false\}$ that satisfies $\Phi$.\\
$\Leftarrow$)
Let $\Gamma$ be a truth assignment that satisfies $\Phi$. Below, we propose a topological additive $2$-numbering $f$ of $D_\Phi$:
\begin{itemize}
\item
  For each $x \in X$, let $f(x^1) = f(x^3) = f(x^4) = f(x^5) = 1$ and
  $f(x^2) = f(u^1) = f(u^2) = f(u^3) = f(u^4) = f(u^5) = f(u^6) = 2$;
  if $\Gamma(x) = true$ then let $f(x) = 1$ and $f(\lnot x) = 2$,
  otherwise, let $f(x) = 2$ and $f(\lnot x) = 1$.
  Then, $S(x^3) = S(x^4) = S(x^5) = 2$, $S(x^2) = 4$, $S(x^1) = 5$, $S(u^1) = S(u^2) = S(u^3) \leq 2$,
  $S(u^4) = S(u^5) = S(u^6) \leq 2$ and for all $x \in X \cup \lnot X$ we have $S(x) \geq 7$.
  Moreover, $S(x) \ge 9$ when $(c, x) \in A$.
\item
  For each $c \in C$, let $f(c) = f(c^2) = f(c^3) = f(c^4) = 2$ and $f(c^1) = f(c^5) = 1$. Then,
  $S(c^2) = S(c^3) = S(c^4) = 1$, $S(c^5) = 2$, $S(c^1) = 8$ and $5 \leq S(c) \leq 7$ (since $\Gamma$ satisfies $\Phi$).
\end{itemize}
$\Rightarrow$)
Let $f$ be a topological additive $2$-numbering $f$ of $D_\Phi$.
By Lemma~\ref{LEMA1O}, for each $x \in X$, the values $f(x)$ and $f(\lnot x)$ cannot be both 1.
Hence, we can set $\Gamma(x) = true$ when $f(x) = 1$ and $\Gamma(x) = false$ when $f(\lnot x) = 1$.
In the case that $f(x) =  f(\lnot x) = 2$, $\Gamma(x)$ may be arbitrarily $true$ or $false$.
Now, by Lemma~\ref{LEMA2O}, for every $c = y \lor z \lor w$, at least one of the three values $f(y)$, $f(z)$, $f(w)$ must be 1.
Therefore, the assignment satisfies $c$ and then $\Phi$.
\end{proof}

\end{document}